\documentclass[11pt]{article}

\usepackage[margin=1in]{geometry}
\usepackage{setspace}
\usepackage{fancyhdr}

\usepackage{amsmath}
\usepackage{amssymb}
\usepackage{amsthm}
\usepackage{mathtools}

\usepackage{graphicx}
\usepackage{booktabs}
\usepackage{tabularx}
\usepackage{array}
\usepackage{makecell}

\usepackage{enumitem}

\newtheorem{proposition}{Proposition}

\newtheorem{definition}{Definition}

\title{Strategic Voting Axioms for Multiwinner Elections}
\author{Austin Lieberman and Matthew I. Jones}
\date{2026}
\begin{document} 


\maketitle

\begin{abstract}
We create a set of 6 axioms relating to strategic voting and study them on single transferable vote (STV), Meek STV, Bloc, Chamberlin-Courant, Monroe, and $k$-Borda. For Chamberlin-Courant, Monroe, and $k$-Borda, we implement ``optimistic'' and ``pessimistic'' models for handing incomplete ballots, thus turning each of the three into two separate voting rules. We determine which voting rules satisfy each axiom. We find that under Chamberlin-Courant and $k$-Borda, the optimistic models encourage burying, while the pessimistic models encourage truncation strategies, indicating that the handling of partial ballots has tremendous implications for strategic voting. We define an axiom called no skipped payment that differentiates similar rules, as Meek STV and Chamberlin-Courant satisfy it, while STV and Monroe fail. Each voting rule was found to fail favorite betrayal, revealing that none of them are completely immune to strategic voting. In general, STV variants are relatively resistant to strategic voting while Monroe is susceptible to a wide range of strategies.
\end{abstract}

\section{Introduction}
\label{sec:introduction}

Many elections require multiple candidates to be selected, including parliamentary elections, creating a shortlist of job candidates or finalists, or electing school boards and city councils. Countries such as Scotland, Australia, and the United States use multiwinner elections for local and federal governments.

A multiwinner voting rule takes in ballots and selects a winning set of candidates, referred to as the committee. We look only at voting rules in which the ballots are ordered rankings of the candidates. There are many such rules, and prior literature has compared them by defining mathematical properties, or axioms, and testing which voting rules satisfy or fail them \cite{elkind,faliszewski2019,graham}. Most of this axiomatic work has been focused on creating properties relating to proportional representation -- the winning committee should proportionally represent the electorate \cite{aziz,elkind}. For example, Elkind et al. \cite{elkind} define a proportionality property, ``solid coalitions'': if at least $n/k$ voters rank a candidate first, where $n$ is the number of voters and $k$ is the size of the committee, that candidate must be elected. Properties like these carry few incentives for voters to order their ballot in a certain way. If a voter knows that a voting rule fails solid coalitions, it is not clear if this should affect how they vote. On the other hand, consider the later-no-harm property, which as been studied on single-winner elections: a winning candidates should not become a loser when a voter appends candidates to their ballot \cite{woodall1997}. If an election rule fails later-no-harm, a voter has a clear incentive to submit a ballot ranking fewer candidates than they would have on an honest ballot. The multiwinner elections literature, especially for voting rules using ordered ballots, is largely lacking such incentive-based axioms. We define a set of axioms that relate to how voters are incentivized to insincerely rank candidates on their ballots in an attempt to elect a more favorable winning committee.

Dating back to Arrow \cite{arrow}, the axiomatic approach often followed in social choice literature aims to evaluate the fairness of rules for turning individual preferences into a group decision. The Gibbard-Satterthwaite theorem \cite{gibbard,satt} states that any non-dictatorial single-winner voting rule is susceptible to some form of strategic voting. Strategic voting is generally at odds with fairness. Dishonest voters may have an advantage over honest ones. As argued by Brandt et al. \cite{brandt}, a voter who has greater resources to identify potentially beneficial strategies has an unfair advantage. Thus, axiomatizing the strategic susceptibility of voting rules is necessary for evaluating their fairness.

These axioms also show how voting rules respond to incomplete or partial ballots. Much of the multiwinner social choice literature assumes that voters have complete preferences over all of the candidates and that every voter ranks every candidate. However, this assumption is unrealistic, especially in elections with many candidates or elections where all candidates may not be well known, like political elections~\cite{graham}. Furthermore, in some jurisdictions, voters only have the option to rank a few candidates \cite{min}. To this end, we adapt common voting rules so that they accept partial ballots, and several of the axioms we create explicitly involve the strategic consequences of submitting partial ballots. Graham-Squire et al. \cite{graham} developed axioms to study the representation implications, but not the strategic implications, of partial ballots in multiwinner election.

The rest of the paper is structured as follows. In Section \ref{sec:preliminaries} all preliminary information is provided and the axioms are formalized. In Section \ref{sec:results} proofs are provided for each voting rule failing or satisfying each axiom, and Section \ref{sec:discussion} contains a discussion of the results. The results are summarized in Table \ref{tab:results}.

\section{Preliminaries}
\label{sec:preliminaries}

We think of a multiwinner voting rule is a function which takes as input $C$, the set of candidates, $P$, the profile of ballots (referred to as the preference profile), and $k$, the size of the committee, and returns $W$, the winning committee, where $W \subseteq C$ and $\lvert W \rvert = k$. $C=\{c_1, ..., c_m\}$, where $m$ is the number of candidates, and $P=(\succ _1,...\succ_n)$, where $n$ is the number of voters in the election. In the case of a tie, we assume the winning committee is chosen randomly. Each $\succ_i$ represents the ballot of voter $i$, which is an ordered ranking of between 1 and $m$ candidates. Voters cannot represent ties on their ballots and may rank as many or as few candidates as they would like. 

We use the notation $\succ_i^H$ to indicate that $\succ_i$ -- the ballot voter $i$ submits -- is identical to voter $i$'s honest ordered preferences of the candidates. When we write just $\succ_i$, without an $H$, that does not necessarily mean that $\succ_i$ is different from $i$'s honest ordered preferences; it just means that it does not matter for our purposes whether $\succ_i$ is honest or manipulated. $\succ_i^H$ can include fewer than $m$ candidates. It is possible, and sometimes likely, that a voter would not have complete preferences over all of the candidates. 

Additional notation we use includes $r_i(c)$, meaning the position voter $i$ ranks $c$ on $\succ_i$ (where the first-ranked candidate's position is 1), $l_i$, meaning the number of candidates ranked on $\succ_i$, $a \succ_i b$, meaning $i$ ranks $a$ above $b$ on $\succ_i$, and $C_i$, meaning the set of candidates on $\succ_i$. We use preference profiles for which we tabulate the committee when voter 1 submits their original ballot, $\succ_1$ or $\succ_1^H$, and then use $\succ_1^*$ to represent voter 1's manipulated ballot. When we do so, the rest of the notation changes accordingly: $r_1^*(c)$, for example, would mean the ranking of $c$ on $\succ_1^*$.

\subsection{Voting Rules}
\label{subsec:votingrules}

In this paper we do not look at approval-based rules, with one exception. We adapt Bloc so that it accepts ordered ballots. Approval rules are a significant family of multiwinner voting methods, but we chose not to include them since most of the voting strategies addressed in this paper work only with ordered ballots. There has been axiomatic working studying the strategyproofness of multiwinner approval-based rules \cite{peters,caragiannis2026}.

We separate the voting rules into two classes: candidate-based rules and committee-based rules. Candidate-based rules elect each candidate separately, while committee-based rules elect the committee as a unit and voters are represented by a single member of the committee. 

\subsubsection{Candidate-based rules}

\paragraph{Bloc} 
The normal implementation of Bloc, or $k$-approval, lets voters indicate up to $k$ candidates they approve of, and the $k$ candidates with the highest approval scores are elected. Our implementation of Bloc instead takes in ranked ballots and the top $k$ ranked candidates each get 1 point. The $k$ candidates with the most points are elected.

\paragraph{Single transferable vote (STV)} STV elects the committee in a multi-round process, repeating until $k$ candidates have been elected. Each round works as follows. If any candidate $c$ has a score greater than or equal to the quota $q= \lfloor \frac{n}{k+1} \rfloor +1$, they are elected and each of $c$'s ballots transfers transfer proportionally to the next-ranked candidates, assuming all elected and eliminated candidates are removed from all ballots. Formally, the value transferred from each ballot can be calculated as $ \frac{(V)(T(c) - q)}{T(c)}$, where $V$ is the value of that ballot held by $c$ and $T(c)$ is $c$'s current vote total. If no candidate reaches quota, the candidate with the lowest vote total is eliminated from the election and each ballot on which they are ranked first transfers to the next-ranked candidate. This process repeats until either $k$ candidates have been elected, or the number of remaining candidates is equal to the number of empty seats in the committee.

There are many small variations in STV, handling various corner cases and minute details, but this description provides enough information for the reader. We use the version of STV implemented in Scotland. Full details can be found in the statute that implemented it \cite{scottish}.

\paragraph{Meek STV}
As with STV, we will not cover all details of Meek's implementation \cite{meek1969}. Full details are provided by Hill et al. \cite{hill1987} and Hill \cite{hill2006}.

Meek STV differs from STV in two ways that are relevant to the properties this paper explores. The first is that elected candidates continue to receive transferred ballots. Every candidate $c$ has a keep weight $w_c$. Let a \textit{hopeful} candidate be one that is neither elected nor eliminated. Hopeful candidates have a keep weight of 1, meaning they keep each ballot at its full value. Eliminated candidates have a weight of 0, keeping no part of any ballot. Elected candidates have a weight between 0 and 1 that is calculated such that by keeping the fraction $w_c$ of each ballot, their vote total is equal to quota. As discussed next, the quota is also recalculated throughout the election and is dependent upon each candidate's keep weight, so there is an iterative process that repeats until each candidate's weight and the quota are at equilibrium.

Suppose $\succ_i= a\succ b \succ c$, $a$ and $b$ are elected, $w_a=.4$, $w_b=.5$, and $c$ is hopeful. Since each ballot has a value of 1, $a$ keeps .4 and passes the remaining .6 to $b$. However, since $b$ can only keep .5 of it, $b$ receives a value of .3 from $\succ_i$, passing the other .3 to $c$. Since $c$ is hopeful, it keeps the full .3.

The second relevant way that Meek differs from STV is that the quota is calculated differently and changes throughout the election. In Meek, $q= \frac{n-excess}{k+1}$. The numerator contains $excess$, meaning excess votes, or the sum of the values from each ballot that do not go to a candidate because no further candidates are available to keep that value. If the example above were changed so that $c$ were elected and $w_c=.5$, then $c$ would keep .15 and the remaining .15 would be considered excess.

\paragraph{Pessimistic $k$-Borda ($k$PM) and Optimistic $k$-Borda ($k$OM)}
The remaining rules are all based on the Borda score. The \emph{Borda score} of a candidate $c$ on ballot $\succ_i$ is defined as $m - r_i(c)$ \cite{elkind}. We use the notation $S(c)$, meaning candidate $c$'s cumulative Borda score across all ballots, and $S_i(c)$, meaning the Borda score $\succ_i$ gives to $c$. For example, if $m=4$ and $\succ_i=a\succ b\succ c \succ d$, then $S_i(a)=3$, $S_i(b)=2$, $S_i(c)=1$, and $S_i(d)=0$.

In Bloc and both STV variants, there is an implied procedure for the handling of partial ballots. In Bloc, a voter's top $k$ candidates receive 1 point while the rest receive 0. In STV, votes go into excess. However, there is no obvious way of dealing with partial ballots in Borda rules. We implement the ``optimistic'' and ``pessimistic'' models of Baumeister et al. \cite{baumeister}, which were originally used for single-winner Borda-count elections and later applied to the Chamberlin-Courant multiwinner voting rule, for example by Graham-Squire et al. \cite{graham}.

A \textbf{pessimistic Borda score function} can be defined as
\[
S ^{PM} _i (c) = 
\begin{cases}
    m-r_i(c) & \text{if voter } i \text{ ranks } c\\
    0 & \text{otherwise}
\end{cases}
\]
That is, the score voter $i$ gives $c$ is $c$'s Borda score based on their ranking, unless $c$ is not ranked, in which case the score is 0.

Alternatively, an \textbf{optimistic Borda score function} can be defined as
\[
S^{OM} _i (c) = 
\begin{cases}
    m-r_i(c) & \text{if voter } i \text{ ranks } c\\
    m-l_i-1 & \text{otherwise}
\end{cases}
\]
The difference here is that if $i$ does not rank $c$, $c$ can receive a score greater than 0. The optimistic model treats unranked candidates as if they were ranked next at the end of $i$'s ballot. In this paper we do not claim one is superior. Rather, we investigate the different strategic behaviors the two models encourage.

The $k$-Borda rule is a straightforward application of the single-winner Borda-count to the multiwinner realm. Here the pessimistic and optimistic variants are defined.

Formally, $c$'s score is
\[
S ^{k \text{PM}} (c) = \sum_{i=1}^n ( S ^{PM} _i (c) )
\qquad
S ^{k \text{OM}} (c) = \sum_{i=1}^n ( S ^{OM} _i (c) )
\]

For $k$PM and $k$OM, the sum of the Borda scores of each candidate across all ballots is calculated, and the $k$ candidates with the highest scores are elected. 

\subsubsection{Committee-based rules}
The next two voting rules, Chamberlin-Courant \cite{chamberlin} and Monroe \cite{monroe}, elect the committee as a unit rather than electing $k$ separate candidates. Voters are assigned representatives from each committee, but the rules differ in how the representatives are chosen.

\paragraph{Pessimistic and optimistic Chamberlin-Courant (CC PM and CC OM)}
For every potential committee, each voter is assigned to its favorite member and derives a utility value calculated as the Borda score of that candidate on their ballot. The committee that maximizes satisfaction, or utility, across all voters is elected. Let $G$ denote a possible committee (any subset of $C$ of size $k$), and let $V(G,c)$ denote the set of voters whose highest-ranked candidate from $G$ is $c$. If a voter ranks no candidates in $G$, they must be placed in $V(G,c)$ for exactly one $c \in G$. Under \textbf{CC PM} and \textbf{CC OM} the total score of $G$ is
\[
S ^{CCPM} (G) = 
\sum _ {c \in G}
\sum _ {i \in V(G,c)} 
\left( S ^{PM} _i (c) \right)
\qquad
S ^{CCOM} (G) = 
\sum _ {c \in G}
\sum _ {i \in V(G,c)}
\left( S^{OM} _i (c) \right)
\]

\paragraph{Pessimistic and optimistic Monroe (Monroe PM and Monroe OM)}
Under Monroe, each candidate from any potential committee $G$ must represent $1/k$ of the voters. Voters are assigned to their representatives in a way that maximizes total satisfaction, which is calculated using the Borda score, just like CC. However, unlike CC, a voter may not be represented by their favorite member of $G$. We avoid Monroe elections that have a number of ballots that cannot be cleanly divided between the candidates, so there is no need to address a procedure for when that is the case. Let $A(G)$ be the set of all possible assignments of voters to representatives in $G$, such that each $c \in G$ represents $n/k$ voters, and let $a(i)$ be voter $i$'s representative. Under \textbf{Monroe PM} and \textbf{Monroe OM} the total score of $G$ is
\[
S ^{Monroe \text{ } PM} (G) = 
\max _ {a \in A(G)}
\sum _ {i = 1} ^n
\left( S ^{PM} _i ( a(i) \right)
\qquad 
S ^{Monroe \text{ } OM} (G) = 
\max _ {a \in A(G)}
\sum _ {i = 1} ^n
\left( S ^{OM} _i ( a(i) \right)
\]
Under both Monroe and Chamberlin-Courant, the committee $G$ with the highest score is elected.

We use the notation $c(G)$ to represent the set of voters represented by candidate $c$ for committee $G$, where $c \in G$. 

Note that in any profile where all ballots are complete, for $k$-Borda, Chamberlin-Courant, and Monroe, the optimistic and pessimistic models will produce the same result since they only differ when there are partial ballots.

\subsection{Voter utility}
\label{subsec:utility}
One of the axioms in the paper is structured as follows: is it possible for a voter to improve their outcome by following this strategy? In order to answer this, we must define what it means for a voter to prefer one committee to another. Suppose there is an election in which $m=6$, $k=2$, and $\succ_1 = a \succ b \succ c \succ d \succ e \succ f$. Would voter $1$ prefer the winning committee to be $\{a,f\}$ or $\{b,d\}$? The only way to determine this truthfully is by asking voter $1$ which outcome they prefer, but the only information given about $1$ is their ordered ranking of the candidates, which may not even be complete. Here we propose two methods for solving this problem.

\paragraph{Best representative utility}
One can determine which committee gives a voter greater utility by comparing the candidate the voter ranks highest on their honest ballot from each committee: a voter is more satisfied with the committee from which they rank a candidate higher. If a voter's highest-ranked candidate is the same across multiple committees, then their second-highest-ranked candidates are compared, and so on. If voters rank no candidates from any of the committees being compared, they are considered indifferent. This is inspired by the Chamberlin-Courant rule which calculates voter satisfaction from their best representative.

\paragraph{Borda score utility}
Another way of determining which committee a voter prefers is to compare the cumulative Borda scores of the candidates in each committee. unranked candidates contribute a score of 0, like the pessimistic method. The voter is more satisfied with the committee that has a higher cumulative Borda score.

\medskip

Now, looking back at our example, $1$ prefers $\{a,f\}$ under best representative utility but prefers $\{b,d\}$ under Borda score utility, as the cumulative Borda score of $\{a,f\}$ is 5 and $\{b,d\}$ is 6.

It happens to be the case that in the entire paper there is no instance in which there is a voter that can improve their outcome under just one of the utility models and not both. Thus, we need not choose between them, or make the argument that one is superior. Accordingly, anytime we write that a voter improves their outcome, their outcome has improved in terms of both utility definitions.

One important detail to clarify is that voter 1's utility is calculated from $\succ_1$, not $\succ_1^*$. A voter's true feelings about the candidates are represented by $\succ_1$, so their satisfaction is based on $\succ_1$.

\subsection{Axioms}
\label{sec:properties}
Formalizations of each axiom, or property, and the strategies they connect to are described in this section.

\begin{definition}
\textbf{Append later-no-harm}. Let $W$ denote the winning committee in an election with ballot $\succ_i$, and let $W^*$ denote the winning committee in an election where all other ballots stay the same but $\succ_i$ changes to $\succ_i^*$. A voting rule satisfies append later-no-harm if, for every preference profile $P$ and every voter $i$, there exists no $\succ_i^*$ such that there is a candidate $c \in (W \setminus W^*) \cap C_i$, given that appended to the end of $\succ_i$ is at least one candidate $a \in C_i^* \setminus C_i$, and when all such candidates $a$ are removed from $\succ_i^*$, $\succ_i=\succ_i^*$.
\end{definition}

Put simply, a candidate that won originally should not become a loser when a voter appends additional candidates to their ballot below that original winner. \textit{Later-no-harm} and \textit{later-no-help} were defined by Woodall \cite{woodall1997} for single-winner elections. He defines later-no-harm as ``Adding a later preference to a ballot should not harm any candidate already listed'' and later-no-help as ``Adding a later preference to a ballot should not help any candidate already listed.'' He wrote these definitions so they would also apply to multiwinner elections. To clarify, Woodall intended, as do we, for the definition to allow for the appending of more than one candidate to a ballot.

If an election rule fails append later-no-harm, voters are encouraged to submit a truncated, or shortened, version of their honest ballot, as ranking many candidates may jeopardize the possibility of their favorite candidates being elected. To demonstrate this, consider this profile under $k$PM which demonstrates its failure of append later-no-harm and the related truncation strategy. As for the notation of preference profiles, each $\succ_i$ represents a ballot, with the highest candidate being ranked first. Let there be an original election with $\succ_1^H$, and an alternative election with $\succ_1^*$ replacing $\succ_1^H$ and all other ballots constant.

\begin{table}[h]
\centering
\setlength{\tabcolsep}{12pt}
\begin{tabular}{cccccc}
\multicolumn{6}{c}{$m=5$, $k=2$} \\ \hline
$\succ_1^H$&  $\succ_1^*$   &$\succ_2$   & $\succ_3$   & $\succ_4$  & $\succ_5$   \\ \hline
$a$ & $a$ & $a$ & $b$ & $c$ & $c$ \\
$b$ & $b$ & $b$ & $d$ & $b$ & $a$ \\
$c$ &     & $c$ &     & $e$ &     \\ \hline
\end{tabular}
\caption{$k$PM fails append later-no-harm.}
\label{tab:kPM_ap_harm}
\end{table}

In the original election in Table \ref{tab:kPM_ap_harm}, with $\succ_1^H$, the scores of each candidate are $a=11$, $b=13$, $c=12$, $d=3$, and $e=2$, so $W=\{b,c\}$. However, in the election with $\succ_1^*$ replacing $\succ_1^H$, $W^*=\{a,b\}$ as $c=10$ while all other candidates' scores stay the same. It is clear that truncating their ballot to exclude $c$ is an effective strategic decision by voter $1$, because according to $\succ_1^H$, they prefer $W^*$ to $W$ (under both utility models).

This example also proves that $k$PM fails append later-no-harm because $a$ is harmed by the appending of $c$ below itself. If a voting rule fails append later-no-harm, it will be susceptible to strategic ballot truncation of this type; voters will be able to cause the election of their highly ranked candidates by leaving lower ones off their ballot.

\begin{definition}
    \textbf{Reorder later-no-harm.} A voting rule satisfies reorder later-no-harm if, for every preference profile $P$ and every voter $i$, there exists no $\succ_i^*$ such that there is a candidate $c \in (W \setminus W^*) \cap C_i$, given that $\succ_i$ and $\succ_i^*$ are identical up to and including $c$, $C_i=C_i^*$, and after $c$ is ranked there exists at least one candidate $a$ such that $a \in C_i$ and $r_i(a) \neq r_i^*(a)$.
\end{definition}

In other words, reordering candidates on a ballot should not harm any candidate ranked above the reordered candidates. This axiom is another subdivision of Woodall's later-no-harm that focuses on a specific form of strategic behavior. 

The effect of strategic truncation is reducing the support for a candidate competing with a voter's more preferred candidates. In Table \ref{tab:kPM_ap_harm}, voter 1 reduced support for $c$ to help $a$. Another strategy that accomplishes the same thing is moving a candidate down on the ballot, rather than removing them. This is demonstrated in Table \ref{tab:kPM_re_harm}, under $k$PM. Note that all ballots are the same as in Table \ref{tab:kPM_ap_harm}, except $\succ_1^H$ and $\succ_1^*$ are now complete.

\begin{table}[h]
\centering
\setlength{\tabcolsep}{12pt}
\begin{tabular}{cccccc}
\multicolumn{6}{c}{$m=5$, $k=2$} \\ \hline
$\succ_1^H$&  $\succ_1^*$   &$\succ_2$   & $\succ_3$   & $\succ_4$  & $\succ_5$   \\ \hline
$a$ & $a$ & $a$ & $b$ & $c$ & $c$ \\
$b$ & $b$ & $b$ & $d$ & $b$ & $a$ \\
$c$ & $d$ & $c$ &     & $e$ &     \\ 
$d$ & $e$ &     &     &     &     \\ 
$e$ & $c$ &     &     &     &     \\ \hline
\end{tabular}
\caption{$k$PM fails reorder later-no-harm.}
\label{tab:kPM_re_harm}
\end{table}

The same mechanism occurs here; $W=\{b,c\}$, but $c$ loses 2 points as voter $1$ lowers them on their ballot, and $W^*=\{a,b\}$. Voter 1 is more satisfied when they submit $\succ_1^*$, making their strategy of lowering $c$ effective. This profile also proves that $k$PM fails reorder later-no-harm, as $a$ is harmed by the reordering of candidates below it. If a voting rule fails reorder later-no-harm, strategic manipulations like this will work.

\begin{definition}
\textbf{Append later-no-help.} A voting rule satisfies append later-no-help if, for every preference profile $P$ and every voter $i$, there exists no $\succ_i^*$ such that there is a candidate $c \in (W^* \setminus W) \cap C_i$, given that appended to the end of $\succ_i$ is at least one candidate $a \in C_i^* \setminus C_i$, and when all such candidates $a$ are removed from $\succ_i^*$, $\succ_i=\succ_i^*$.
\end{definition}

Simply, a candidate that lost originally should not become a winner when a voter appends additional candidates to their ballot below that original loser. To illustrate the strategy that is made possible by the failure of append later-no-help, consider the $k$OM election in Table \ref{tab:kOM_ap_help}.

\begin{table}[h]
\centering
\setlength{\tabcolsep}{12pt}
\begin{tabular}{ccccc}
\multicolumn{5}{c}{$m=5$, $k=2$} \\ \hline
$\succ_1^H$& $\succ_1^*$  &  $\succ_2$   & $\succ_3$  & $\succ_4$  \\ \hline
$a$ & $a$ &  $a$   & $b$ & $e$ \\
$b$ & $b$ &  $b$   & $e$ & $c$ \\
    & $c$ &  $e$   & $a$ & $b$  \\ 
    & $d$ &  $d$   & $d$ & $d$ \\
    & $e$ &  $c$   & $c$ & $a$  \\\hline
\end{tabular}
\caption{$k$OM fails append later-no-help.}
\label{tab:kOM_ap_help}
\end{table}

In the original election, with $\succ_1^H$, $W=\{b,e\}$ as $a=10$, $b=12$, $c=5$, $d=5$, and $e=11$. However, in the alternative election, with $\succ_1^*$, $W^*=\{a,b\}$ as $a=10$, $b=12$, $c=5$, $d=4$, and $e=9$. Voter 1 is more satisfied with $W^*$, making their strategic manipulation in $\succ_1^*$ effective. This is a strategy called burying: lowering the ranking of one candidate to help another \cite{green}. Here, $e$ is buried to help $a$. 

Additionally, this profile proves that $k$OM fails append later-no-help as $a$ is helped by the appending of additional candidates below it. Voting rules that fail append later-no-help are susceptible to burying.

\begin{definition}
    \textbf{Reorder later-no-help.} A voting rule satisfies reorder later-no-help if, for every preference profile $P$ and every voter $i$, there exists no $\succ_i^*$ such that there is a candidate $c \in (W^* \setminus W) \cap C_i$, given that $\succ_i$ and $\succ_i^*$ are identical up to and including $c$, $C_i=C_i^*$, and after $c$ is ranked there exists at least one candidate $a$ such that $a \in C_i$ and $r_i(a) \neq r_i^*(a)$.
\end{definition}

Put simply, reordering candidates on a ballot should not help any candidate ranked above the reordered candidates. Burying also works by reordering. Suppose in Table \ref{tab:kOM_ap_help} that $\succ_1^H=a \succ b \succ e \succ c \succ d$, and $\succ_1^*$ is the same as it is in Table \ref{tab:kOM_ap_help}. The outcome in the $\succ_1^H$ election and the $\succ_1^*$ election are exactly the same as they are in the original profile. Voter $1$ achieves the same outcome by burying $e$ through reordering rather than appending candidates, as $e$ goes from being ranked $3^{\text{rd}}$ to $5^{\text{th}}$. If a rule fails reorder later-no-help, burying by reordering will work.

\begin{definition}
    \textbf{No skipped payment -- Candidate-based rules.} 
    Let $W$ be the winning committee. A candidate-based voting rule satisfies no skipped payment if, for every preference profile $P$ and every voter $i$, for every pair of candidates $a,b$, given that $a\in W $ and $a\succ_i b$, if $i$ paid for $b$, then $i$ also paid for $a$. Voter $i$ is considered to have paid for candidate $c$ if $c$ is ranked on $\succ_i$, and

    \begin{itemize}
        \item (for STV rules) at least some fraction of $i$'s ballot is kept by $c$.

        \item (for $k$-Borda) $c$ is not the $m^{\text{th}}$ ranked candidate on $\succ_i$.

        \item (for Bloc) $c$ is ranked inside the top $k$ on $\succ_i$.
    \end{itemize}
\end{definition}

\begin{definition}
    \textbf{No skipped payment -- Committee-based rules.} 
    Let $W$ be the winning committee. A committee-based voting rule satisfies no skipped payment if, for every preference profile $P$ and every voter $i$, for every pair of candidates $a,b$, given that $a\in W $ and $a\succ_i b$, there is no committee containing both $a$ and $b$ from which $i$ is represented by $b$.
\end{definition}

Simply, if a voter contributes to a lower-ranked candidate, they must have also contributed to each higher-ranked candidate that is elected. We should note that we create this property specifically to outline differences between Meek and STV, and between Monroe and CC. The concept of payment does not fit particularly well for $k$-Borda and Bloc as voters contribute to every candidate they rank under $k$-Borda, with the exception of the $m^\text{th}$ ranked candidate, and similarly under Bloc they contribute to every candidate they rank in their top $k$ (there is no point in ranking more than $k$ candidates under Bloc, so we could say voters effectively contribute to every candidate they rank. We only give voters the option to rank more than $k$ to stay consistent with the rest of the rules). If a voter ranks all $m$ candidates, it is unfair to say that they paid for their least favorite candidate, but beyond that it is arbitrary to determine which ranked candidates a voter did or did not pay for under Bloc and $k$-Borda. Additionally, voters will contribute to every member of $W$ they rank under $k$-Borda, with the exception of the $m^\text{th}$ ranked candidate, and every member in the top $k$ under Bloc, but under the STV rules, Monroe, and CC this is not always true. We choose to include Bloc and $k$-Borda anyway so that the results are complete.

We create this axiom in response to a specific strategy in STV, first introduced by Woodall \cite{woodall1982}, and later named \textit{Woodall free riding} \cite{schulze}. In most regular variants of STV, an elected candidate cannot receive additional votes. These votes skip right over them, and go to the next hopeful candidate. Woodall free riding takes advantage of this in STV and is demonstrated in Table \ref{tab:scotSTV_skipped}.

\begin{table}[h]
\centering
\setlength{\tabcolsep}{12pt}
\begin{tabular}{cccccc}
\multicolumn{6}{c}{$m=4$, $k=2$} \\ \hline
$\succ_1^H$& $\succ_1^*$  &  $\succ_2$   & $\succ_{3,5}$  & $\succ_6$ & $\succ_{7,8}$ \\ \hline
$a$ & $d$ &  $a$  &  $a$  & $b$ & $c$ \\
$b$ & $a$ &  $b$  &       &     &     \\
$d$ & $b$ &       &       &     &     \\ \hline
\end{tabular}
\caption{STV fails no skipped payment. The notation $\succ_{3,5}$ means voters 3 through 5 all submit identical ballots.}
\label{tab:scotSTV_skipped}
\end{table}

Since there are 8 voters and $k=2$, the quota is 3. When voter 1 submits $\succ_1^H$, at the initial count $a=5$, $b=1$, $c=2$, and $d=0$. $a$ is elected, and a total of .8 votes transfer to $b$ as surplus. Thus, $b=1.8$, $c=2$, and $d=0$. $d$ is eliminated, $c$ is elected, and $W=\{a,c\}$. Alternatively, in the election with $\succ_1^*$, at the initial count $a=4$, $b=1$, $c=2$, and $d=1$. $a$ is elected and .25 votes transfer to $b$. Now, $b=1.25$, $c=2$, and $d=1$. $d$ is eliminated and 1 whole vote transfers to $b$, skipping $a$ on $\succ_1^*$. $b=2.25$ and $c=2$, so $b$ is elected and $W^*=\{a,b\}$. Voter $1$ is taking advantage of elected candidates not receiving transferred votes in STV. In the original election, only a fraction of 1's vote transfers to $b$ because it is coming from a  surplus. However, by putting a weak candidate $d$ above $a$, once $d$ has been eliminated and $a$ has already been elected, the whole vote skips $a$, and goes to $b$. Voter 1 pays for $b$ but skips paying for $a$, thus improving their outcome. This example shows that STV fails no skipped payment. This free riding strategy does not work in Meek STV because $a$ would keep part of 1's ballot after $d$ is eliminated (although in Meek, the winning committee is $\{a,b\}$ in both elections).

While this property is created for Woodall free riding, an STV-specific strategy, skipped payment can also happen with the committee-based rules. In these rules, voters only pay for one member of each committee -- their representative. In CC, this is always their highest-ranked member of the committee, but in Monroe it may not be. A voter can take advantage of this by following the strategy demonstrated in Table \ref{tab:MPM_MOM_skipped}, which works under optimistic and pessimistic Monroe as the profile contains only full ballots.

\begin{table}[h]
\centering
\setlength{\tabcolsep}{12pt}
\begin{tabular}{ccccc}
\multicolumn{5}{c}{$m=4$, $k=2$} \\ \hline
$\succ_1^H$& $\succ_1^*$  &  $\succ_2$   & $\succ_3$  & $\succ_4$\\ \hline
$a$ & $b$ &  $a$   & $a$ & $c$\\
$b$ & $a$ &  $d$   & $c$ & $b$ \\
$c$ & $c$ &  $c$   & $d$ & $d$\\ 
$d$ & $d$ &  $b$   & $b$ & $a$\\\hline
\end{tabular}
\caption{Monroe PM and OM fail no skipped payment.}
\label{tab:MPM_MOM_skipped}
\end{table}

In the original election, the highest-scoring committees are $\{a,c\}=11$ and $\{a,b\}=10$, so $\{a,c\}$ wins. $a(\{a,c\})=\{1,2\}$, meaning voters 1 and 2 are represented by $a$ for the committee $\{a,c\}$, $c(\{a,c\})=\{3,4\}$, $a(\{a,b\})=\{2,3\}$, and $b(\{a,b\})=\{1,4\}$. Since each candidate must represent an equal number of voters, the score of $\{a,b\}$ is maximized by having voter $1$ be represented by $b$. This makes sense as $a$ is ranked first by three voters, yet can only represent two. Since 1 ranks $b$ second, they are the best of the three voters to assign to $b$. If 1 knows they are likely to be assigned to $b$ for $\{a,b\}$, then they could increase $\{a,b\}$'s chances of winning by raising $b$ on their ballot as they do on $\succ_1^*$. In the election with $\succ_1^*$, $\{a,c\}=10$, $\{a,b\}=11$, and 1's representative from $\{a,b\}$ is still $b$. Thus, $\{a,b\}$ wins. Voter 1 improves their outcome, making their strategy successful. In the original election, Monroe fails no skipped payment, since voter 1 pays for $b$ yet not $a$, as $b$ is voter 1's representative for $\{a,b\}$, while $a\succ_1^Hb$.

\begin{definition}
    \textbf{Favorite betrayal.} Let $a$ be the first-ranked candidate on $\succ_i^H$, and let $b$ be the first-ranked candidate on $\succ_i^*$, where $a \neq b$. Let $W$ be the winning committee in the election with $\succ_i^H$, and let $W^*$ be the winning committee in the election with $\succ_i^*$. A voting rule satisfies favorite betrayal if, for every preference profile $P$ and every voter $i$, $i$ does not prefer $W^*$ to $W$ when $C_i=C_i^*$, and the only difference between $\succ_i^H$ and $\succ_i^*$ is that the positions of $a$ and $b$ are swapped.
\end{definition}

Simply, no voter should be able to improve their outcome by swapping a candidate with their honest favorite. The related strategy is straightforward: give additional support to a candidate by moving them to the very top of the ballot. Table \ref{tab:MPM_MOM_skipped} demonstrates an example of this as voter $1$ improves their outcome by moving $b$ to the top of their ballot.

Favorite betrayal has been studied on single-winner rules, and few satisfy it, especially ones with ranked ballots of the nature this paper looks at \cite{wolk,small}. With that in mind, we did not expect to find a multiwinner voting rule that satisfies it, and we did not. Authors normally define favorite betrayal such that voters can manipulate their ballot in any way, as long as they move a candidate above their favorite. This paper's definition is more restrictive, allowing voters to swap a candidate with their favorite and perform no further manipulations. Because every rule we study fails this more restricted definition, they also fail the standard definition.

We considered creating additional betrayal properties that are geared specifically to multiwinner elections. One idea was a property that states no voter should be able to reorder their top $k$ candidates and improve their outcome, and a second idea states no voter should be able to move a candidate from outside their top $k$ to inside their top $k$ and improve their outcome. We tested both and none of the voting rules satisfied either so we chose not to formally include them.

\section{Results}
\label{sec:results}

\begin{table}[h]
\centering

\renewcommand{\arraystretch}{1.5}
\setlength{\tabcolsep}{4pt}

\begin{tabularx}{\textwidth}{
    l
    >{\centering\arraybackslash}X
    >{\centering\arraybackslash}X
    >{\centering\arraybackslash}X
    >{\centering\arraybackslash}X
    >{\centering\arraybackslash}X
    >{\centering\arraybackslash}X
}
Rule
 & \makecell{Ap. later-\\no-harm}
 & \makecell{Re. later-\\no-harm}
 & \makecell{Ap. later-\\no-help}
 & \makecell{Re. later-\\no-help}
 & \makecell{No skipped\\payment}
 & \makecell{Favorite\\betrayal} \\ \hline

STV  & $\checkmark$ & $\checkmark$ & $\checkmark$ & $\checkmark$ & $\times$ & $\times$ \\
Meek STV      & $\checkmark$ & $\checkmark$ & $\checkmark$ & $\checkmark$ & $\checkmark$ & $\times$ \\
$k$PM         & $\times$ & $\times$ & $\checkmark$ & $\times$ & $\checkmark$ & $\times$ \\
$k$OM         & $\checkmark$ & $\times$ & $\times$ & $\times$ & $\checkmark$ & $\times$ \\
CC PM         & $\times$ & $\times$ & $\checkmark$ & $\times$ & $\checkmark$ & $\times$ \\
CC OM         & $\checkmark$ & $\times$ & $\times$ & $\times$ & $\checkmark$ & $\times$ \\
Monroe PM     & $\times$ & $\times$ & $\times$ & $\times$ & $\times$ & $\times$ \\
Monroe OM     & $\times$ & $\times$ & $\times$ & $\times$ & $\times$ & $\times$ \\ 
Bloc          & $\times$! & $\times$! & $\checkmark$ & $\times$! & $\checkmark$ & $\times$! \\ \hline

\end{tabularx}

\caption{Summary of results. $\checkmark$ means the voting rule satisfies the axiom and $\times$ means it fails. ! means that the rule fails only with a tie.}
\label{tab:results}
\end{table}

Before going through each property and proving whether the voting rules satisfy or fail them, we introduce Proposition \ref{prop:1}, which applies to several proofs.

\begin{proposition}
    Under STV and Meek STV, candidates are unaffected by ballot manipulations that occur below them. \label{prop:1}
\end{proposition}

\begin{proof}
    Suppose there is a ballot $\succ_i$ and a candidate $c \in C_i$. According to the definitions of both STV rules, the ballot or any fraction of the ballot can only transfer to the candidate ranked below $c$ once $c$ has been elected or eliminated. Since the ballot only reaches candidates ranked below $c$ once $c$'s fate has already been decided, any ballot manipulation strictly below $c$ cannot affect $c$'s outcome. Note that under Meek $c$'s score can change after they are elected, but they will always remain elected.
\end{proof}

\subsection{Append later-no-harm}

\begin{proposition}
    STV, Meek STV, CC OM, and $k$OM satisfy append later-no-harm.
\end{proposition}

\begin{proof}
    By Proposition 1, since candidates in STV and Meek STV are unaffected by ballot manipulations that occur strictly below them, these candidates can never be harmed by the appending of additional candidates below them. Therefore both rules satisfy append later-no-harm.

    Now consider CC OM. Let $P$ be any preference profile, let $i$ be a voter in $P$ with the ballot $\succ_i$, and let $W$ be the winning committee of the election with $\succ_i$. A voting rule can only fail append later-no-harm in an election in which $i$ ranks a winning candidate. Otherwise, there is no candidate to harm. Let $\succ_i^*$ be voter $i$'s ballot after appending at least one candidate to the end, and let $a$ be any appended candidate. Let $X$ be any committee containing a candidate ranked by $i$ and let $Y$ be any committee that contains no such candidate. 
    
    For every $X$, appending additional candidates to $\succ_i$ will never change $i$'s contribution to $X$'s score because under CC a voter is represented by their highest-ranked member of the committee. For every $Y$, appending additional candidates to $\succ_i$ will never increase $i$'s contribution to $Y$'s score. Originally, $S_i(Y)=m-l_i-1$ points. But, if none of the appended candidates belong to $Y$, $Y$ may now receive fewer points as $S_i(Y)=m-l_i^*-1$ where $l_i^*>l_i$, or if one of the appended candidates $a$ belongs to $Y$, $Y$ may now receive the same or fewer points as $S_i(Y)=m-r_i(a)$, where $r_i(a) \geq l_i-1$.

    If $i$ ranks a member of $W$, $W$ belongs to the class of committees $X$. Since the scores of all committees $X$ are constant, and the scores of all committees $Y$ are constant or decrease, $W$ must stay as the winner in the election with $\succ_i^*$. Therefore, no winning candidate can be harmed by the appending of additional candidates below it, so CC OM satisfies later-no-harm.

    As for $k$OM, the score of a candidate $c$ ranked on $\succ_i$ will be unaffected by the appending of additional candidates as $S_i(c)$ is $c$'s Borda score. An appended candidate $a$ will either have their score stay the same or decrease as $S_i(a)=m-r_i(a)$, where $r_i(a) \geq l_i-1$. Lastly, a candidate $d$ not ranked on $\succ_i$ or $\succ_i^*$ will have their score decrease as $S_i(d)=m-l_i^*-1$ where $l_i^*>l_i$. Since the score of a winning candidate on $\succ_i$ will never change, and the scores of all losing candidates will never increase, a voter will never cause the loss of an originally winning candidate by appending additional candidates below them. Therefore $k$OM satisfies append later-no-harm.
    
\end{proof}

\begin{proposition}
    $k$PM, CC PM, Monroe PM, Monroe OM, and Bloc fail append later-no-harm.
\end{proposition}

\begin{proof}
\begin{table}[h]
\centering
\setlength{\tabcolsep}{6pt}

\begin{tabular}{ccccc||ccccccc||cccc}

\multicolumn{5}{c||}{$m=4$, $k=2$}
&
\multicolumn{7}{c||}{$m=5$, $k=2$}
&
\multicolumn{4}{c}{$m=3$, $k=2$}
\\ \hline

$\succ_{1,5}$ & $\succ_1^*$ & $\succ_{6,9}$ & $\succ_{10,15}$ & $\succ_{16}$
&
$\succ_1$ & $\succ_1^*$ & $\succ_2$ & $\succ_3$ & $\succ_4$ & $\succ_5$ & $\succ_6$
&
$\succ_{1,3}$ & $\succ_1^*$ & $\succ_{4,5}$ & $\succ_{6,9}$
\\ \hline

$a$ & $a$ & $b$ & $c$ & $d$
&
$a$ & $a$ & $c$ & $a$ & $e$ & $b$ & $e$
& 
$a$ & $a$ & $b$ & $c$
\\

& $b$ & & & $b$
& 
& $b$ & $a$ & $c$ & $c$ & $a$ & $b$
&
& $b$ & &
\\

& & & &
&
& $c$ & $d$ & $b$ & $d$ & $d$ & $d$
&
& & &
\\

& & & &
&
& & $b$ & & & $c$ &
&
& & &
\\ \hline

\end{tabular}

\caption{Profiles proving failure of append later-no-harm. Left is CC PM and Monroe PM, middle is Monroe OM, and right is Bloc.}
\label{tab:ap_harm}
\end{table}

A profile proving $k$PM fails append later-no-harm is provided in Table \ref{tab:kPM_ap_harm}.

In Table \ref{tab:ap_harm}, in the left profile under both CC PM and Monroe PM, in the original election, $\{a,c\}=33$ and $\{b,c\}=32$. Therefore $W=\{a,c\}$. Under Monroe PM the representatives are $a(\{a,c\})=\{1,\ldots,8 \}$, $c(\{a,c\})=\{9, \ldots,16 \}$, $b(\{b,c\})=\{1, 2, 3, 6, \ldots, 9, 16\}$, and $c(\{b,c\})=\{4, 5, 10, \ldots , 15 \}$. In the election with $\succ_1^*$, under CC PM and Monroe PM, $\{a,c\}=33$ and $\{b,c\}=34$. Therefore $W^*=\{b,c\}$. The representatives do not change under Monroe PM. $a$ is harmed by voter 1 appending $b$. 

In the middle profile under Monroe OM, $\{a,e\}=21$ and $\{b,c\}=20$. Therefore $W=\{a,e\}$. The representatives are $a(\{a,e\})=\{2,3,5\}$, $e(\{a,e\})=\{1,4,6\}$, $b(\{b,c\})=\{1,5,6\}$, and $c(\{b,c\})=\{2,3,4\}$. In the election with $\succ_1^*$, $\{a,e\}=19$ and $\{b,c\}=20$. Therefore $W^*=\{b,c\}$. The representatives do not change. $a$ is harmed by voter 1 appending $b$ and $c$. 

In the right profile under Bloc, $W=\{a,c\}$ while there is a tie between $\{a,c\}$ and $\{b,c\}$ for $W^*$. Assuming random tie-breaking, there is a 50\% chance $a$ is harmed by appending $b$. 
\end{proof}

\subsection{Reorder later-no-harm}

\begin{proposition}
    STV and Meek STV satisfy reorder later-no-harm.
\end{proposition}
\begin{proof}
    By Proposition 1, candidates can never be affected, and thus harmed, by the reordering of candidates below them in STV and Meek STV. Therefore both rules satisfy reorder later-no-harm.
\end{proof}

\begin{proposition}
    $k$PM, $k$OM, CC PM, CC OM, Monroe PM, Monroe OM, and Bloc fail reorder later-no-harm.
\end{proposition}
\begin{proof}
    
\begin{table}[h]
\centering
\setlength{\tabcolsep}{4pt}

\begin{tabular}{cccccc||ccccc||cccccc}

\multicolumn{6}{c||}{$m=6$, $k=2$}
&
\multicolumn{5}{c||}{$m=5$, $k=2$}
&
\multicolumn{6}{c}{$m=6$, $k=3$}
\\ \hline

$\succ_1$ & $\succ_1^*$ & $\succ_2$ & $\succ_{3,12}$ & $\succ_{13,17}$ & $\succ_{18,27}$
&
$\succ_1$ & $\succ_1^*$ & $\succ_2$ & $\succ_3$ & $\succ_4$
&
$\succ_1$ & $\succ_1^*$ & $\succ_{2,4}$ & $\succ_{5,7}$ & $\succ_{8,12}$ & $\succ_{13,17}$
\\ \hline

$a$&$a$ &$f$ &$a$ &$e$ &$b$
&
$a$&$a$ &$a$ &$b$ &$e$
&
$a$&$a$ &$a$ &$d$ &$e$ &$f$
\\

$b$&$b$ &$e$ &$f$ &$d$ &$f$
&
$b$&$b$ &$e$ &$a$ &$b$
&
$b$& $b$& & & &
\\

$c$&$f$ &$c$ &$c$ &$c$ &$d$
&
$c$&$e$ &$d$ &$e$ &$d$
&
$c$&$d$ & & & &
\\

$d$&$d$ &$b$ &$d$ &$f$ &$a$
&
$d$&$d$ &$c$ &$c$ &$c$
&
$d$&$c$ & & & &
\\

$e$&$e$ &$a$ &$b$ &$a$ &$c$
&
$e$&$c$ &$b$ &$d$ &$a$
&
& & & & &
\\ 

$f$&$c$ &$d$ &$e$ &$b$ &$e$
&
&&&&
&
& & & & &
\\ \hline

\end{tabular}

\caption{Profiles proving failure of reorder later-no-harm. Left is CC PM and CC OM, middle is $k$PM and $k$OM, right is Bloc.}
\label{tab:re_harm_1}
\end{table}

\begin{table}[h]
\centering
\setlength{\tabcolsep}{10pt}

\begin{tabular}{ccccccccccc}
\multicolumn{11}{c}{$m=5$, $k=2$} \\ \hline

$\succ_1$ & $\succ_1^*$ & $\succ_{2,3}$ & $\succ_{4,5}$ & $\succ_{6,7}$
& $\succ_{8,9}$ & $\succ_{10,11}$ & $\succ_{12,13}$ & $\succ_{14,15}$ & $\succ_{16,17}$ & $\succ_{18}$ \\ \hline

$a$&$a$ &$a$ &$a$ &$b$ &$b$ &$d$ &$d$ &$e$ &$e$ &$e$ \\
$b$&$d$ &$d$ &$e$ &$d$ &$e$ &$a$ &$b$ &$a$ &$b$ &$c$ \\
$c$&$c$ &$b$ &$b$ &$a$ &$a$ &$e$ &$e$ &$d$ &$d$ &$b$ \\
$d$&$b$ &$e$ &$d$ &$e$ &$d$ &$b$ &$a$ &$b$ &$a$ &$a$ \\
$e$&$e$ &$c$ &$c$ &$c$ &$c$ &$c$ &$c$ &$c$ &$c$ &$d$ \\ \hline

\end{tabular}

\caption{Profile proving Monroe PM and Monroe OM fail reorder later-no-harm.}
\label{tab:re_harm_2}
\end{table}

In Table \ref{tab:re_harm_1}, in the left profile under both CC PM and CC OM, $\{a,b\}=112$ and $\{e,f\}=111$. Therefore $W=\{a,b\}$. In the election with $\succ_1^*$, $\{a,b\}=112$ and $\{e,f\}=113$. Therefore $W^*=\{e,f\}$. $a$ is harmed when voter 1 reorders candidates below them. 

In the middle profile under $k$PM and $k$OM, $W=\{a,b\}$ as $a=11$, $b= 10$, $c=5$, $d=5$, and $e=9$. $W^*=\{a,e\}$ as $a=11$, $b= 10$, $c=3$, $d=5$, and $e=11$. Thus, $b$ is harmed by the reordering of candidates below it. 

In the right profile under Bloc, $W=\{a,e,f\}$ while there is a tie between $\{a,e,f\}$ and $\{d,e,f\}$ for $W^*$. There is a 50\% chance $a$ is harmed. In Table \ref{tab:re_harm_2}, under both Monroe PM and Monroe OM, $\{a,b\}=62$ and $\{d,e\}=61$. Therefore $W=\{a,b\}$. The representatives are $a(\{a,b\})=\{1,\ldots, 5,10,11,14,15 \}$, $b(\{a,b\})=\{ 6, \ldots , 9,12,13,16,17,18 \}$, $d(\{d,e\})=\{ 1,2,3,6,7,10,\ldots,13 \}$, and $e(\{d,e\})=\{ 4,5,8,9,14,\ldots,18 \}$. In the election with $\succ_1^*$, $\{a,b\}=62$ and $\{d,e\}=63$. Therefore $W^*=\{d,e\}$. The representatives do not change. $a$ and $b$ are harmed by the reordering of candidates below them.
\end{proof}

\subsection{Append later-no-help}

\begin{proposition}
    STV, Meek STV, $k$PM, CC PM, and Bloc satisfy append later-no-help.
\end{proposition}

\begin{proof}
    By Proposition 1, candidates can never be helped by the appending of candidates below them under STV and Meek STV. Therefore both satisfy append later-no-help.

    Let $P$ be any preference profile, $i$ be a voter in $P$ with the ballot $\succ_i$, and let $W$ be the winning committee of the election with $\succ_i$. Let $\succ_i^*$ be voter $i$'s ballot after appending at least one candidate to the end, and let $a$ be any appended candidate. 
    
    Consider CC PM. Let $X$ be any committee with a candidate in $C_i$, and let $Y$ be any committee that contains no such candidate. For each possible $X$, $S_i(X)=S_i^*(X)$ since $X$'s score is determined by their best representative, which does not change when candidates are appended to $\succ_i$. For each $Y$, $S_i(Y) \leq S_i^*(Y)$ as $S_i(Y)=0$, but $S_i^*(Y)>0$ could be true if there exists a candidate in $Y \cap C_i^* \setminus C_i$. A candidate on $\succ_i$ cannot be helped by appending additional candidates to $\succ_i$ since the scores of the committees containing them, $X$, do not change, while the scores of the other committees, $Y$, increase or stay the same. Therefore CC PM satisfies append later-no-help.

    Now, consider $k$PM and Bloc. Let the notation $S_i(c)$ also be used here for Bloc to represent the number of points voter $i$ gives to candidate $c$. For a candidate $c$ on $\succ_i$, $S_i(c)=S_i^*(c)$, and for each candidate $a$, $S_i(a) \leq S_i^*(a)$. Each candidate $c$ cannot be helped by appending candidates to $\succ_i$ since $c$'s score stays the same, while the other candidates' scores either increase or stay the same. Therefore $k$PM and Bloc satisfy append later-no-help.
    
\end{proof}

\begin{proposition}
    $k$OM, CC OM, Monroe PM, and Monroe OM fail append later-no-help.
\end{proposition}
\begin{proof}
\begin{table}[h]
\centering
\setlength{\tabcolsep}{3pt}

\begin{tabular}{ccccccccccc||ccccccc}

\multicolumn{11}{c||}{$m=5$, $k=2$}
&
\multicolumn{7}{c}{$m=5$, $k=2$}
\\ \hline

$\succ_1$ & $\succ_1^*$ & $\succ_{2,3}$ & $\succ_{4,5}$ & $\succ_{6,7}$
& $\succ_{8,9}$ & $\succ_{10,11}$ & $\succ_{12,13}$ & $\succ_{14,15}$ & $\succ_{16,17}$ & $\succ_{18}$
&
$\succ_1$ & $\succ_1^*$ & $\succ_2$ & $\succ_3$ & $\succ_4$ & $\succ_5$ & $\succ_6$
\\ \hline

$a$&$a$ &$a$ &$a$ &$b$ &$b$ &$d$ &$d$ &$e$ &$e$ &$e$
&
$b$& $b$ & $c$ &$d$ &$a$ &$e$ &$d$
\\

&$b$ &$d$ &$e$ &$d$ &$e$ &$a$ &$b$ &$a$ &$b$ &$c$
&
$e$&$e$ &$e$ &$b$ &$e$ &$a$ &$a$
\\

&$c$ & & & & & & & & &
&
&$d$ &$b$ & &$b$ &$c$ &$e$
\\

&$d$ & & & & & & & & &
&
& & & & &$d$ &$b$
\\

&$e$ & & & & & & & & &
&
& & & & & &
\\ \hline

\end{tabular}

\caption{Profiles proving failure of append later-no-help. Left is CC OM and Monroe OM, and right is Monroe PM.}
\label{tab:ap_help}
\end{table}

A profile proving $k$OM fails append later-no-help is provided in Table \ref{tab:kOM_ap_help}.

In Table \ref{tab:ap_help}, in the left profile under CC OM and Monroe OM, $\{d,e\}=63$ and $\{a,b\}=62$. Therefore $W=\{d,e\}$. Under Monroe OM the representatives are $a(\{a,b\})= \{ 1,\ldots,5,10,11,14,15  \}$, $b(\{a,b\})= \{ 6,\ldots,9,12,13,16,17,18 \}$, $d(\{d,e\})= \{ 1,2,3,6,7,10,\ldots,13 \}$, and $e(\{d,e\})= \{ 4,5,8, $ $9,14,\ldots,18\}$. In the election with $\succ_1^*$, $\{d,e\}=61$ and $\{a,b\}=62$. Therefore $W^*=\{a,b\}$. The representatives do not change under Monroe OM. $a$ is helped by voter 1 appending candidates below it. 

In the right profile under Monroe PM, $\{a,b\}=19$ and $\{d,e\}=18$. Therefore $W=\{a,b\}$. The representatives are $a(\{a,b\}) = \{ 4,5,6 \}$, $b(\{a,b\}) = \{ 1,2,3 \}$, $d(\{d,e\}) = \{ 1,3,6 \}$, and $e(\{d,e\}) = \{ 2,4,5 \}$. In the election with $\succ_1^*$, $\{a,b\}=19$ and $\{d,e\}=20$. Therefore $W^*=\{d,e\}$. The representatives do not change. $e$ is helped by voter 1 appending $d$ below it.
\end{proof}

\subsection{Reorder later-no-help}

\begin{proposition}
    STV and Meek STV satisfy reorder later-no-help.
\end{proposition}

\begin{proof}
    By Proposition 1, candidates are unaffected and therefore not helped by the reordering of candidates below them under STV and Meek STV. Therefore both of these rules satisfy reorder later-no-help.
\end{proof}

\begin{proposition}
    $k$PM, $k$OM, CC PM, CC OM, Monroe PM, Monroe OM, and Bloc fail reorder later-no-help.
\end{proposition}
\begin{proof}
\begin{table}[h]
\centering
\setlength{\tabcolsep}{4pt}

\begin{tabular}{ccccc||cccccc||cccccc}

\multicolumn{5}{c||}{$m=5$, $k=2$}
&
\multicolumn{6}{c||}{$m=6$, $k=2$}
&
\multicolumn{6}{c}{$m=6$, $k=3$}
\\ \hline

$\succ_1$ & $\succ_1^*$ & $\succ_2$ & $\succ_3$ & $\succ_4$
&
$\succ_1$ & $\succ_1^*$ & $\succ_2$ & $\succ_{3,12}$ & $\succ_{13,22}$ & $\succ_{23,32}$
&
$\succ_1$ & $\succ_1^*$ & $\succ_{2,3}$ & $\succ_{4,6}$ & $\succ_{7,12}$ & $\succ_{13,18}$
\\ \hline

$a$&$a$ &$a$ &$b$ &$c$
&
$a$&$a$ &$f$ &$c$ &$f$ &$a$
&
$a$&$a$ &$a$ &$c$ &$e$ &$f$
\\

$b$&$b$ &$c$ &$a$ &$b$
&
$b$&$b$ &$d$ &$b$ &$b$ &$d$
&
$b$&$b$ & & & &
\\

$c$&$e$ &$e$ &$c$ &$e$
&
$c$&$e$ &$c$ &$e$ &$e$ &$f$
&
$c$&$d$ & & & &
\\

$d$&$d$ &$d$ &$d$ &$a$
&
$d$&$d$ &$b$ &$d$ &$c$ &$e$
&
$d$&$c$ & & & &
\\

$e$&$c$ &$b$ &$e$ &$d$
&
$e$&$c$ &$a$ &$a$ &$d$ &$c$
&
& & & & &
\\ \hline

\end{tabular}

\caption{Profile proving failure of reorder later-no-help. Left is $k$OM and $k$PM, middle is CC PM and CC OM, and right is Bloc.}
\label{tab:re_help_1}
\end{table}
\begin{table}[h]
\centering
\setlength{\tabcolsep}{10pt}

\begin{tabular}{ccccccccccc}
\multicolumn{11}{c}{$m=5$, $k=2$} \\ \hline

$\succ_1$ & $\succ_1^*$ & $\succ_{2,3}$ & $\succ_{4,5}$ & $\succ_{6,7}$
& $\succ_{8,9}$ & $\succ_{10,11}$ & $\succ_{12,13}$ & $\succ_{14,15}$ & $\succ_{16,17}$ & $\succ_{18}$ \\ \hline

$a$&$a$ &$a$ &$a$ &$b$ &$b$ &$d$ &$d$ &$e$ &$e$ &$e$ \\
$d$&$b$ &$d$ &$e$ &$d$ &$e$ &$a$ &$b$ &$a$ &$b$ &$c$ \\
$e$&$c$ &$b$ &$b$ &$a$ &$a$ &$e$ &$e$ &$d$ &$d$ &$b$ \\
$b$&$d$ &$e$ &$d$ &$e$ &$d$ &$b$ &$a$ &$b$ &$a$ &$d$ \\
$c$&$e$ &$c$ &$c$ &$c$ &$c$ &$c$ &$c$ &$c$ &$c$ &$a$ \\ \hline

\end{tabular}

\caption{Profile proving Monroe PM and Monroe OM fail reorder later-no-help.}
\label{tab:re_help_2}
\end{table}

In Table \ref{tab:re_help_1}, in the left profile under $k$OM and $k$PM, $W=\{a,c\}$ as $a=12$, $b= 10$, $c=11$, $d=3$, and $e=4$. $W^*=\{a,b\}$ as $a=12$, $b= 10$, $c=9$, $d=3$, and $e=6$. Thus, $b$ is helped by the reordering of candidates below it. 

In the middle profile under CC PM and CC OM, $\{c,f\}=138$ and $\{a,b\}=137$. Therefore $W=\{c,f\}$. In the election with $\succ_1^*$, $\{c,f\}=136$ and $\{a,b\}=137$. Therefore $W^*=\{a,b\}$. $a$ and $b$ are helped by the reordering of candidates below them. 

In the right profile under Bloc, $W=\{c,e,f\}$ while there is a tie between $\{a,e,f\}$ and $\{c,e,f\}$ for $W^*$. There is a 50\% chance $a$ is helped by the reordering of candidates below it. 

In Table \ref{tab:re_help_2}, under Monroe PM and Monroe OM, $\{d,e\}=63$ and $\{a,b\}=62$. Therefore $W=\{d,e\}$. The representatives are $a(\{a,b\})=\{1,\ldots, 5,10,11,14,15 \}$, $b(\{a,b\})=\{ 6, \ldots , 9,12,13,16,$ $17,18 \}$, $d(\{d,e\})=\{ 1,2,3,6,7,10,\ldots,13 \}$, and $e(\{d,e\})=\{ 4,5,8,9,14,\ldots,18 \}$. In the election with $\succ_1^*$, $\{a,b\}=62$ and $\{d,e\}=61$. Therefore $W^*=\{a,b\}$. The representatives do not change. $a$ is helped by the reordering of candidates below it. 
\end{proof}

\subsection{No skipped payment}

\begin{proposition}
    Meek STV, CC PM, CC OM, $k$OM, $k$PM, and Bloc satisfy no skipped payment.
\end{proposition}
\begin{proof}
    By the definition of Meek STV, every candidate $c \in W$ has a keep weight $w_c>0$. Let us say a candidate has been \textit{revealed} on a ballot if all candidates above them have been elected or eliminated. For each pair of candidates $a,b \in C_i$, such that $a\in W$ and $a\succ_ib$, in order for $i$ to have paid for $b$, $a$ must also have been revealed as $a\succ_ib$. Since $a$ is revealed on $\succ_i$ and $w_a>0$, $i$ paid for $a$. Thus, if $i$ paid for $b$, then $i$ must have also paid for $a$. Therefore Meek STV satisfies no skipped payment.

    Under both CC variants, a voter is represented by their highest-ranked member of each committee, assuming they rank at least one. Therefore, for each pair of candidates $a,b \in C_i$ that are in a given committee $G$, such that $a\in W$ and $a\succ_ib$, it is impossible for $i$'s representative from $G$ to be $b$ and therefore both CC PM and CC OM satisfy no skipped payment. 

    Under Bloc, a voter pays for all candidates $c$ such that $r_i(c) \leq k$, and pays only for these candidates. It is impossible for a voter to have paid for $b$ and not $a$ where $a \succ_i b$, and therefore Bloc satisfies no skipped payment.

    Under $k$-Borda, a voter pays for every candidate in their top $m-1$ that is elected and therefore cannot skip paying for one. Thus, it satisfies no skipped payment.
\end{proof}

\begin{proposition}
    STV, Monroe PM, and Monroe OM fail no skipped payment.
\end{proposition}
\begin{proof}
    A profile proving STV fails no skipped payment is provided in Table \ref{tab:scotSTV_skipped}, and Monroe PM and Monroe OM in Table \ref{tab:MPM_MOM_skipped}.
\end{proof}

\subsection{Favorite betrayal}
\begin{proposition}
    Every voting rule we consider fails favorite betrayal.
\end{proposition}

\begin{proof}
\begin{table}[h]
\centering
\setlength{\tabcolsep}{8pt}

\begin{tabular}{cccccccc||ccccc}

\multicolumn{8}{c||}{$m=5$, $k=2$}
&
\multicolumn{5}{c}{$m=4$, $k=2$}
\\ \hline

$\succ_1$ & $\succ_1^*$ & $\succ_2$ & $\succ_3$ & $\succ_{4,7}$ & $\succ_{8,11}$ & $\succ_{12,15}$ & $\succ_{16,19}$ 
&
$\succ_1$ & $\succ_1^*$ & $\succ_{2,5}$ & $\succ_{6,7}$ & $\succ_{8,10}$ 
\\ \hline

$a$&$c$ &$c$ &$e$ &$d$ &$d$ &$e$ &$c$
&
$a$&$c$ &$b$ &$c$ &$d$
\\

$b$&$b$ &$e$ &$a$ &$e$ &$c$ &$d$ &$d$
&
$b$&$b$ & & &
\\

$c$&$a$ &$b$ &$b$ &$c$ &$e$ &$c$ &$e$
&
$c$&$a$ & & &
\\

$d$&$d$ &$a$ &$d$ &$a$ &$b$ &$a$ &$b$
&
$d$&$d$ & & &
\\

$e$&$e$ &$d$ &$c$ &$b$ &$a$ &$b$ &$a$
&
& & & &
\\ \hline

\end{tabular}

\caption{Profiles proving failure of favorite betrayal. Left is CC PM and CC OM and right is Bloc.}
\label{tab:fb_1}
\end{table}

\begin{table}[h]
\centering
\setlength{\tabcolsep}{4pt}
\begin{tabular}{cccccc||cccccc}
\multicolumn{6}{c||}{$m=4$, $k=2$} 
&
\multicolumn{6}{c}{$m=4$, $k=2$} 
\\ \hline
$\succ_{1,10}$ & $\succ_1^*$ & $\succ_{11,60}$ & $\succ_{61,106}$ & $\succ_{107,134}$ & $\succ_{135,154}$ 
&
$\succ_{1,2}$ & $\succ_1^*$ & $\succ_{3,12}$ & $\succ_{13,21}$ & $\succ_{22,27}$ & $\succ_{28,31}$
\\ \hline
$a$ & $b$ &  $a$  &  $b$  & $c$ & $d$ 
&
$a$&$b$ & $a$& $b$&$c$ & $d$ \\
$b$ & $a$ &       &       &     & $c$ 
& $b$&$a$ & & & & $c$ \\
    &     &       &       &     &     
    &
    & & & & & \\ \hline
\end{tabular}
\caption{Left is STV and right is Meek STV.}
\label{tab:fb_scot}
\end{table}

\begin{table}[h]
\centering
\setlength{\tabcolsep}{2.5pt}

\begin{tabular}{ccccccccc||ccccccccc}

\multicolumn{9}{c||}{$m=5$, $k=2$}
&
\multicolumn{9}{c}{$m=5$, $k=2$}
\\ \hline

$\succ_1$ & $\succ_1^*$ & $\succ_{2,3}$ & $\succ_{4,5}$ & $\succ_{6,7}$ & $\succ_{8,9}$ & $\succ_{10,11}$ & $\succ_{12,13}$ & $\succ_{14}$
&
$\succ_1$ & $\succ_1^*$ & $\succ_{2,3}$ & $\succ_{4,5}$ & $\succ_{6,11}$ & $\succ_{12,17}$ & $\succ_{18,19}$ & $\succ_{20,21}$ & $\succ_{22}$
\\ \hline

$a$&$c$ &$c$ &$c$ &$d$ &$d$ &$e$ &$e$ &$d$
&
$a$&$c$ &$c$ &$c$ &$d$ &$d$ &$e$ &$e$ &$e$
\\

$b$&$b$ &$d$ &$e$ &$c$ &$e$ &$d$ &$c$ &$e$
&
$b$&$b$ &$d$ &$e$ &$c$ &$e$ &$d$ &$c$ &$a$
\\

$c$&$a$ &$e$ &$d$ &$e$ &$c$ &$c$ &$d$ &$a$
&
$c$&$a$ &$e$ &$d$ &$e$ &$c$ &$c$ &$d$ &$d$
\\

$d$&$d$ &$a$ &$b$ &$a$ &$b$ &$a$ &$b$ &$b$
&
$d$&$d$ &$a$ &$b$ &$a$ &$b$ &$a$ &$b$ &$b$
\\

$e$&$e$ &$b$ &$a$ &$b$ &$a$ &$b$ &$a$ &$c$
&
$e$&$e$ &$b$ &$a$ &$b$ &$a$ &$b$ &$a$ &$c$
\\ \hline

\end{tabular}

\caption{Left is $k$PM and $k$OM, and right is Monroe PM and Monroe OM.}
\label{tab:fb_2}
\end{table}

In Table \ref{tab:fb_scot}, in the left profile for STV and in the right profile for Meek STV, $W=\{a,c\}$ while $W^*=\{a,b\}$. Voter 1's outcome improves by swapping $b$ with their favorite candidate $a$ in both profiles.

In the middle profile under CC PM and CC OM, $\{d,e\}=68$ and $\{c,d\}=67$. Therefore $W=\{d,e\}$. In the election with $\succ_1^*$, $\{d,e\}=68$ and $\{c,d\}=69$. Therefore $W^* =\{c,d\}$. Voter 1's outcome improves by swapping $c$ with their favorite. 

In the right profile under Bloc, $W=\{b,d\}$ while there is a tie between $\{b,d\}$ and $\{b,c\}$ for $W^*$. There is a 50\% chance that voter 1 improves their outcome by swapping $c$ with their favorite. 

In Table \ref{tab:fb_2}, in the left profile under $k$PM and $k$OM, $W=\{d,e\}$ as $a=12$, $b= 10$, $c=38$, $d=41$, and $e=39$. $W^*=\{c,d\}$ as $a=10$, $b= 10$, $c=40$, $d=41$, and $e=39$. Voter 1's outcome improves by swapping $c$ with their favorite. 

In the right profile under Monroe PM and Monroe OM, $\{d,e\}=77$ and $\{c,d\}=76$. Therefore $W=\{d,e\}$. The representatives are $c(\{c,d\})=\{ 1,\ldots,9,20,21 \}$, $d(\{c,d\})=\{ 10,\ldots,19,22 \}$, $d(\{d,e\})=\{ 1,2,3,6,\ldots,13 \}$, and $e(\{d,e\})=\{ 4,5,14,\ldots,22 \}$. In the election with $\succ_1^*$, $\{d,e\}=77$ and $\{c,d\}=78$. Therefore $W^*=\{c,d\}$. Voter 1's outcome improves by swapping $c$ with their favorite. 
\end{proof}

\section{Discussion}
\label{sec:discussion}
In this paper we aim to address two gaps in the multiwinner social choice literature: the lack of axioms relating to voter strategy and a lack of work that allows for incomplete preferences. To both ends, we obtain significant and novel results.

\textbf{How incomplete ballots are handled has tremendous implications for voter strategy.} For each of the three rules based on the Borda score, we use optimistic and pessimistic models of handling partial ballots. Unranked candidates (under $k$-Borda) or unranked committees (under Monroe and CC) are given 0 points under the pessimistic model, but are given points as if the voters were to have ranked them in the highest empty position on their ballot under the optimistic model. As the results illustrate, these two models allow for massively different strategies. $k$PM and CC PM fail append later-no-harm and satisfy append later-no-help, while $k$OM and CC OM fail append later-no-help and satisfy append later-no-harm. The pessimistic models encourage truncation strategies, which makes sense as voters can give candidates or committees they dislike 0 points by leaving them unranked. Meanwhile, the optimistic models encourage burying strategies; the fewer candidates a voter ranks, the more points they give to unranked and undesired candidates or committees. Thus, by submitting longer, dishonest ballots, voters give the lower-ranked candidates or committees fewer points than they would have with a short ballot.

In real elections, especially those with lots of candidates or alternatives, it is unrealistic to expect voters to have complete preferences. Yet, multiwinner social choice literature frequently does not consider this fact, assuming complete preferences. These results make clear the need for more work analyzing multiwinner elections with incomplete ballots, as results can differ greatly depending on how these ballots are treated.

\textbf{No skipped payment differentiates similar voting rules.}  To our knowledge, no skipped payment is the first axiom proved to separate STV from Meek STV, as Meek STV satisfies it while STV fails. In both rules, voters have limited currency. Their ballot has a value of 1, and they would like to spend it in a way that optimizes their outcome. STV fails no skipped payment, meaning a voter can avoid paying for a winning candidate higher on their ballot so they can spend more on a lower ranked one.  In Meek STV, which satisfies no skipped payment, this does not work. By failing no skipped payment, under STV voters can manipulate their preferences to be skipped, thus saving their entire vote for down-ballot candidates.

Additionally, no skipped payment separates the two committee-based rules, as CC satisfies it while Monroe fails. CC and Monroe are quite similar, with the only difference being that under CC each voter is represented by their favorite candidate from the committee while under Monroe each candidate from the committee must represent an equal number of voters. By failing no skipped payment, under Monroe if a voter suspects that one of their top candidates will be skipped, they could insincerely move a candidate up on their ballot to improve their outcome.

\textbf{STV Meek STV perform very well on later-no-harm and later-no-help properties.} Both voting rules satisfy all four of the later-no-harm and later-no-help properties, being the only rules to satisfy more than one of them. Thus, the two rules are immune to the truncation and burying strategies afforded by the failure of the later-no-harm and later-no-help axioms. We should note that STV is still susceptible to some less intuitive truncation strategies that do not involve removing lower-ranked candidates to help a higher one \cite{graham2}. 

\textbf{No rule we study satisfies favorite betrayal.} We expected this to be true prior to proving it, given that nearly all single-winner election rules fail it. This aligns with the Gibbard–Satterthwaite theorem \cite{gibbard,satt}, as we find that none of the voting rules we study are fully strategyproof.

\textbf{These axioms should be understood within the larger context of each voting rule.} Just because one voting rule performs worse on these strategic axioms does not mean it is an objectively worse voting rule overall. For example, $k$-Borda performs worse than STV on this paper's axioms, but $k$-Borda performs better under monotonicity axioms \cite{elkind}. Elkind et al. \cite{elkind} and Faliszewski et al. \cite{faliszewski2017} note that an election rule having a certain property can be of elevated or lowered importance depending upon the nature of the election. If this thinking is applied to these strategic axioms, we believe that satisfying append later-no-help is more important for elections in which the voters may be unfamiliar with some of the candidates. In these elections, it is very natural and likely that voters submit partial ballots that are honest. If an election rule that fails append later-no-help is used for one of these elections, then voters will have lots of opportunities to bury candidates they do not like. Contrastingly, consider an election in which voters are likely to submit honest full ballots. An example would be critics voting for a shortlist of finalists for an award: experts who are well informed of all candidates. It is particularly important that this voting rule satisfies append later-no-harm as voters should not be harming their favorite candidates by honestly ranking lots of candidates.

\textbf{Future work.} These results are purely theoretical; it would be interesting to see how frequently voters succeed when following certain strategies across these election rules. Similarly, it would be interesting to see how frequently following these strategies backfires. Even if a voting rule is susceptible to a strategy, when voters try to follow it, they may actually worsen their outcome. A game-theoretic analysis in which multiple voters or blocs of voters, each with little information, try to select the best strategy could offer valuable insights. Lastly, more work in multiwinner social choice should allow for incomplete preferences.


\newpage

\bibliography{references.bib}

\end{document}